\documentclass[conference]{IEEEtran}
\usepackage{amsmath}
\usepackage{multicol}
\usepackage{mathtools}
\usepackage{amssymb}
\usepackage{amsmath}  
\usepackage{subeqnarray}

\usepackage{algorithm}
\usepackage{algorithmic}        

\usepackage{cite}

\usepackage[dvips]{graphicx}
\usepackage{subfigure}
\usepackage{array}  
\usepackage{url}

\usepackage{amsthm}      
\theoremstyle{plain}     
\newtheorem{thm}{Theorem}
\newtheorem{lem}{Lemma}

\theoremstyle{definition}

\theoremstyle{remark}

\begin{document}
\title{Energy-aware Adaptive Spectrum Access and Power Allocation in LAA Networks via Lyapunov Optimization}

%
\author{\IEEEauthorblockN{Yu Gu,
Qimei Cui,
Yue Wang,
Somayeh Soleimani
}
\IEEEauthorblockA{National Engineering Laboratory for Mobile Network\\
Beijing University of Posts and Telecommunications, Beijing, 100876, China \\
Email: cuiqimei@bupt.edu.cn, guyu@bupt.edu.cn}}
\maketitle

\begin{abstract}
To relieve the traffic burden and improve the system capacity, licensed-assisted access (LAA) has been becoming a promising technology to the supplementary utilization of the unlicensed spectrum. However, due to the densification of small base stations (SBSs) and the dynamic variety of the number of Wi-Fi nodes in the overlapping areas, the licensed channel interference and the unlicensed channel collision could seriously influence the Quality of Service (QoS) and the energy consumption. In this paper, jointly considering time-variant wireless channel conditions, dynamic traffic loads, and random numbers of Wi-Fi nodes, we address an adaptive spectrum access and power allocation problem that enables minimizing the system power consumption under a certain queue stability constraint in the LAA-enabled SBSs and Wi-Fi networks. The complex stochastic optimization problem is rewritten as the difference of two convex (D.C.) program in the framework of Lyapunov optimization, thus developing an online energy-aware optimal algorithm. We also characterize the performance bounds of the proposed algorithm with a tradeoff of $[O(1/V), O(V)]$ between power consumption and delay theoretically. The numerical results verify the tradeoff and show that our scheme can reduce the power consumption over the existing scheme by up to 72.1\% under the same traffic delay.
\end{abstract}

\section{Introduction}
Due to the explosive growth of mobile data stemming from the increasingly prevalence of smart handset devices, the scarcity of spectrum is becoming the bottleneck to boost more capacity of wireless communication \cite{caiwideband}. To improve the system capacity, a common trend has emerged with deploying additional low power nodes (LPNs, such as smallcells, femtocells), and improving the spectral utilization, such as Coordinated Multipoint (CoMP) \cite{Cui2017}. To fundamentally break through this predicament, an emerging technology using the unlicensed spectrum, called licensed-assisted access (LAA), has been launched into the standardization by Third Generation Partnership Project (3GPP) \cite{3GPP}.

There are three major challenges arising in the coexistence networks of LAA-enabled small base stations (SBSs) and Wi-Fi. The first challenge is how to guarantee the fair and effective coexistence between SBSs and WiFi. Due to the time-variant wireless channel conditions and the dynamic variety of the number of Wi-Fi nodes in the overlapping areas, SBS needs a dynamic mechanism to leverage the traffic between the licensed and unlicensed bands \cite{Cui2014}. Secondly, the random arrived traffic and the random access mechanism of LAA become a obstacle to guarantee QoS, which plays an important role in 5G networks. Finally, the new LAA procedures could also have impacts on energy consumption of SBSs due to the extra energy used for channel detection and packet collision.

As for the coexistence of SBSs and Wi-Fi, two kinds of specifications are proposed: frame-based mechanism (FBM) where SBS is activated at periodic cycles on unlicensed band, and load-based mechanism (LBM) where SBS competes for the unlicensed channel using listen-before-talk (LBT) and backoff procedure like Wi-Fi \cite{3GPP,ETSI}. \cite{Zhang2015a,Al-Dulaimi2015,Ratasuk2014} design coexistence mechanisms, such as an almost blank sub-frame (ABS) scheme, an interference avoidance scheme \cite{Zhang2015a}, and adaptive listen-before-talk (LBT) mechanism \cite{Zhang2015a,Ratasuk2014}. To improve the system throughput, \cite{Rupasinghe2015} proposes a Q-Learning based dynamic duty cycle selection technique for configuring LTE transmission gaps.

A few number of works have studied on QoS or energy efficiency (EE) requirements of SBS in the unlicensed band to data. \cite{Yin2016} designs an adaptive adjustment of backoff window size of LAA to minimize the collision probability of Wi-Fi users, satisfying the rate requirements of small cell users. \cite{Yin2016a} develops a power allocation algorithm to obtain pareto optimal between minimization of interference in the licensed band and collision in the unlicensed band, while satisfying the rate requirements of users. \cite{Chen2016} first investigates joint licensed and unlicensed resource allocations to maximize the EE through Nash bargaining when LAA systems adopt a FBM method.

However, \cite{Zhang2015a,Al-Dulaimi2015,Ratasuk2014,Rupasinghe2015,Almeida2013,Yin2016,Yin2016a,Chen2016} focus on static network models and do not fully consider time-varying environment. And most of works ignore the delay impact of LAA network. Therefore, this paper mainly investigates an energy-aware adaptive spectrum access and power allocation problem in coexistence of LAA-enabled SBSs and Wi-Fi networks, hinging on dynamic network model that reflects real network conditions. The main contributions of this paper are threefold.
\begin{itemize}
  \item We address an adaptive spectrum access and power allocation problem that enables minimizing the system average power consumption under a certain queue stability constraint in the LAA-enabled SBSs and Wi-Fi networks, in which the time-variant wireless channel conditions, dynamic traffic loads, and random numbers of Wi-Fi nodes are jointly considered.
  \item The stochastic optimization problem is rewritten as the difference of two convex (D.C.) program, and solved by using the successive convex approximation method in the framework of Lyapunov optimization, thus developing an online energy-aware optimal algorithm.
  \item The theoretical analysis and simulation results show that tuning the control parameter $V$ can quantitatively achieve a tradeoff of $[O(1/V),O(V)]$ between power consumption and delay. The proposed algorithm can reduce the power consumption over the existing scheme by up to 72.1\% under the same traffic delay.
\end{itemize}

The rest of the paper is organized as follows. In Section II, we introduce the system model. In Section III and Section IV, a stochastic optimization problem is formulated and an online energy-aware algorithm is developed based on the Lyapunov optimization. Finally, the numerical results are presented in Section V, and conclusions are given in Section VI.

\section{System Model}
\begin{figure}[!t]
\centering
\includegraphics[width=3.5in]{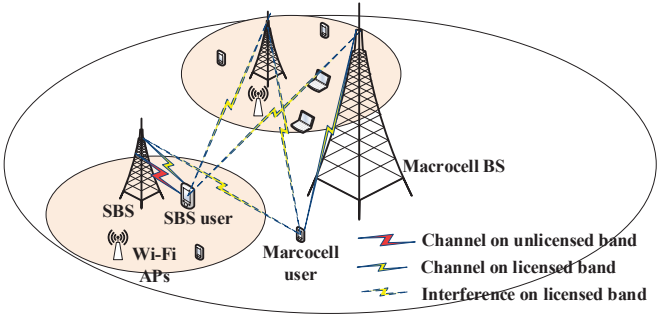}
\caption{System model for SBSs and Wi-Fi coexistence.}
\label{fig.1}
\end{figure}

Consider the downlink of a two-tier wireless network in a slotted system, indexed by  $t \in \{ 0,1,2,...\}$, in which $K$ SBSs share the licensed spectrum with one existing macrocell, and contend the available unlicensed spectrum with Wi-Fi nodes (i.e., Wi-Fi APs, Wi-Fi stations) by using LBT. Denote the set of BSs as $\mathcal{K} = \{ 0,1,2,...,K\}$. Without loss of generality, the marcocell BS is indexed by $0$ and SBSs by $1,2,...,K$. We assume that each SBS works on non-overlapping unlicensed channel. Thus, there is no interference among the SBSs in the unlicensed band. Nevertheless, in the coverage of $k$-th SBS, there are ${N_k}(t)$ Wi-Fi nodes at $t$-th time slot, contending the unlicensed band with $k$-th SBS. With ${N_k}(t)$ varying, the unlicensed band experiences various collisions.

There are ${S_k}$ cellular users in the $k$-th SBS, where $\mathcal{S}_k = \{ 1,2,...S_k\} $  collects the indexes of the users. Further, data packets arrive randomly in every slot and are queued separately for transmission to each user. Let $\mathbf{Q}(t) = \{ {Q_{{s_k}}}(t),\forall {s_k} \in {\mathcal{S}_k},\forall k \in \mathcal{K}\} $ be the queue length vector, where ${Q_{{s_k}}}(t)$ is the queue length of user ${s_k}$ at slot $t$. Let $\mathbf{A}(t) = \{ {A_{{s_k}}}(t),\forall {s_k} \in {\mathcal{S}_k},\forall k \in \mathcal{K}\} $ be the arrival data length vector, where ${A_{{s_k}}}(t)$ is the new traffic arrival amount of user ${s_k}$ at slot $t$. The queues $\mathbf{Q}(t)$ are assumed to be initially empty.

Let $\mathcal{L} = \{ 1,2,...L\}$ and $\mathcal{W} = \{ 1,2,...W\}$ collect the indexes of all the licensed and unlicensed OFDM subcarriers, respectively. We denote the bandwidth of each subcarrier as ${B}$. We denote the licensed and unlicensed subcarrier assignment indicator variables as $x_c^{(k,l,{s_k})}(t)$ and $x_u^{(k,w,{s_k})}(t)$, respectively. Let $p_c^{(k,l,{s_k})}(t)$ and $g_c^{(k,l,{s_k})}(t)$ be the transmit power and the channel gain form the  $k$-th SBS to ${s_k}$-th user on licensed subcarrier $l$ at slot $t$, respectively. Let $p_u^{(k,w,{s_k})}(t)$ and $g_u^{(k,w,{s_k})}(t)$ be the transmit power and the channel gain form the $k$-th SBS to $s_k$-th user on unlicensed subcarrier $w$ at slot $t$, respectively. Denote ${{\bf{x}}_c}(t){\rm{ = (}}x_c^{(k,l,{s_k})}(t){\rm{)}}$, ${{\bf{x}}_u}(t){\rm{ = (}}x_u^{(k,w,{s_k})}(t){\rm{)}}$, and ${\bf{x}}(t) = [{{\bf{x}}_c}(t),{{\bf{x}}_u}(t)]$.  Denote ${{\bf{p}}_c}(t){\rm{ = (}}p_c^{(k,l,{s_k})}(t){\rm{)}}$, ${{\bf{p}}_u}(t){\rm{ = (}}p_u^{(k,w,{s_k})}(t){\rm{)}}$, and ${\bf{p}}(t) = [{{\bf{p}}_c}(t),{{\bf{p}}_u}(t)]$.

\subsection{Transmission rate and power consumption on the licensed band}
The achievable transmission rate of user $s_k$ on the licensed subcarrier $l$ at SBS $k$ at slot $t$, can be given by
\begin{equation}\label{1}
\begin{array}{l}
R_c^{(k,l,{s_k})}(t)\\
= B{\log _2}\left( {1 + \frac{{x_c^{(k,l,{s_k})}(t)p_c^{(k,l,{s_k})}(t)g_c^{(k,l,{s_k})}(t)}}{{\sum\limits_{j \ne k} {x_c^{(j,l,{s_j})}(t)p_c^{(j,l,{s_j})}(t)g_c^{(j,l,{s_k})}(t)}  + {\sigma ^2}}}} \right)
\end{array},
\end{equation}
where ${\sigma ^2}$ is the additive white Gaussian noise (AWGN) power. Meanwhile, it is noteworthy that we need to guarantee the rate of Macrocell's users by imposing a threshold on the cross-tier interference ${I_M}$, which is given as follows
\begin{equation}\label{interference}
\sum\limits_{k \ne 0} {x_c^{(k,l,{s_k})}(t)p_c^{(k,l,{s_k})}(t)g_c^{(k,l,{s_0})}(t)}  \le {I_M}.
\end{equation}

And the transmission power consumption of SBS $k$ on licensed band is
\begin{equation}\label{2}
PC_c^{(k)}(t) = {\xi _c}\sum\limits_{l \in \mathcal{L}} {\sum\limits_{{{s_k}} \in \mathcal{S}_k} {x_c^{(k,l,{s_k})}(t)p_c^{(k,l,{s_k})}(t)} },
\end{equation}
where ${\xi _c}$ is a constant that accounts for the inefficiency of the power amplifiers on licensed band \cite{Li2016}.

\subsection{Transmission rate and power consumption on the unlicensed band}
To guarantee the coexistence with Wi-Fi systems, we assume that SBS adopts an adaptive backoff scheme to access the unlicensed channel, like Wi-Fi. The $k$-th SBS has a attempt transmission probability ${{\tau}_{l,k}}$ and a collision probability ${p_{l,k}}$. All the Wi-Fi nodes within the coverage of the $k$-th SBS are assumed to experience a same attempt transmission probability ${{\tau}_{w,k}}$ and a collision probability ${p_{w,k}}$ in the time slot. The attempt probability of Wi-Fi nodes for given collision probability ${{p}_{w,k}}$ is given by \cite{Kumar2007}
\begin{equation}\label{3}
{{\tau }_{w,k}}(t)=\frac{1+{{p}_{w,k}}+\cdots +p_{w,k}^{{{K}_{w}}-1}}{{{b}_{0}}+{{p}_{w,k}}{{b}_{1}}+\cdots +p_{w,k}^{K{}_{w}-1}{{b}_{{{K}_{w}}-1}}},
\end{equation}
where ${{b}_{j}}$ is the mean backoff time of stage $j$ and ${{K}_{w}}$ is the maximum number of retransmissions for Wi-Fi. The attempt probability of SBSs on unlicensed band is
\begin{equation}
{{\tau }_{l,k}}(t)=\frac{1+{{p}_{l,k}}+\cdots +p_{l,k}^{{{K}_{l}}-1}}{{{e}_{0}}+{{p}_{l,k}}{{e}_{1}}+\cdots +p_{l,k}^{{{K}_{l}}-1}{{e}_{{{K}_{l}}-1}}},
\end{equation}
where ${{e}_{j}}$ is the mean backoff time of stage $j$ and ${{K}_{l}}$ is the maximum number of retransmissions for Wi-Fi.
With the slotted model for the backoff process and the decoupling assumption \cite{Kumar2007}, the collision probabilities of SBSs and WiFi nodes are expressed by respectively
\begin{equation}
{{p}_{w,k}}(t)=1-{{(1-{{\tau }_{w,k}}(t))}^{{{N}_{k}}(t)-1}}(1-{{\tau }_{l,k}}(t)),
\end{equation}
\begin{equation}
{{p}_{l,k}}(t)=1-{{(1-{{\tau }_{w,k}}(t))}^{{{N}_{k}}(t)}}.
\end{equation}
According to Brouwer's fixed point theorem \cite{Kumar2007}, there exists a fixed point for the equations (4)-(7). Hence, we can obtain the attempt transmission probability and the collision probability of SBS and Wi-Fi nodes, respectively.

Then, the successful transmission probability for the $k$-th SBS on unlicensed channel can be given by
\begin{equation}
P_{suc}^{(k)}(t)={{\tau }_{l}}(t){{(1-{{\tau }_{w}}(t))}^{{{N}_{k}}(t)}}.
\end{equation}
Since the time slot of one LTE frame (i.e., 10 ms) is much larger than the Wi-Fi time slot (in the order of $\mu \text{s}$), the time fraction occupied by the SBS on unlicensed channel can be represented by ${{P}^{(k)}_{suc}}(t)$ \cite{Yin2016a}.

Therefore, the achievable transmission rate for user ${{s}_{k}}$ at SBS $k$ on the $w$-th unlicensed subcarrier can be written as
\begin{equation} \label{4}
\begin{array}{l}
R_u^{(k,w,{s_k})}(t)\\
 = P_{suc}^{(k)}(t)B{\log _2}(1 + \frac{{x_u^{(k,w,{s_k})}(t)p_u^{(k,w,{s_k})}(t)g_u^{(k,w,{s_k})}(t)}}{{{\sigma ^2}}}).
\end{array}
\end{equation}

And, the transmission power consumption of SBS $k$ on the unlicensed subcarrier is given by
\begin{equation}
PC_u^{(k)}(t)={\xi _u}\left( {\sum\limits_{w \in \mathcal{W}} {\sum\limits_{{{s_k}} \in \mathcal{S}_k} {x_u^{(k,w,{s_k})}(t)p_u^{(k,w,{s_k})}(t)} } } \right),
\end{equation}
where ${{\xi }_{u}}$ is a constant that accounts for the inefficiency of the power amplifiers on unlicensed band.

\subsection{Total Transmission rate and power consumption of SBSs}
According to (\ref{1}) and (\ref{4}), the achievable transmission data rate for user ${{s}_{k}}$ at SBS $k$ is given by
\begin{equation}
{R^{(k,{s_k})}}(t) = \sum\limits_{l \in \mathcal{L}} {R_c^{(k,l,{s_k})}(t)}  + \sum\limits_{w \in \mathcal{W}} {R_u^{(k,w,{s_k})}(t)}.
\end{equation}

The total transmit rate and the power consumption of SBSs are represented by respectively
\begin{equation}
{R_{tot}}(t) = \sum\limits_{k \in {\cal K}\backslash \{ 0\} } {\sum\limits_{{s_k}} {{R^{(k,{s_k})}}(t)} },
\end{equation}
\begin{equation}
{{PC}_{tot}}(t)=\sum\limits_{k \in \mathcal{K}\backslash \{ 0\}}{\left( P{{C}_{\text{static}}}+PC_{c}^{(k)}(t)+PC_{u}^{(k)}(t) \right)},
\end{equation}
where $P{{C}_{\text{static}}}$ is the static power, consisting of baseband signal processing and additional circuit blocks. Furthermore, we define the average power consumption and the transmit rate of the entire system as
\begin{equation}
{{\overline{PC}}_{tot}}=\underset{t\to \infty }{\mathop{\lim }}\,\frac{1}{t}\sum\limits_{\tau =0}^{t-1}{\mathbb{E}\{P{{C}_{tot}}(\tau )\}},
\end{equation}
\begin{equation}
{{\bar R}_{tot}} = \mathop {\lim }\limits_{t \to \infty } {\mkern 1mu} \frac{1}{t}\sum\limits_{\tau  = 0}^{t - 1} {{\mathbb{E}}\{ {R_{tot}}(\tau )\} }.
\end{equation}
\section{Problem Formulation}
In this section, we process to a stochastic optimization problem to minimize the average power consumption of SBSs, by joint optimizing the licensed and unlicensed subcarriers and power. To guarantee all arrived data leaving the buffer in a finite time, we introduce a concept of queue stability.

The data queue ${{Q}_{{{s}_{k}}}}(t)$ is given by
\begin{equation} \label{queue}
{{Q}_{{{s}_{k}}}}(t+1)\text{=}\max [{{Q}_{{{s}_{k}}}}(t)-{{R}^{(k,{{s}_{k}})}}(t),0]+{{A}_{{{s}_{k}}}}(t),
\end{equation}

And, a queue ${{Q}_{{{s}_{k}}}}(t)$ is strongly stable \cite{neely2010stochastic} if
\begin{equation}
{{\bar{Q}}_{{{s}_{k}}}}=\underset{t\to \infty }{\mathop{\lim }}\,\frac{1}{t}\sum\limits_{\tau =0}^{t-1}{E\{\left| {{Q}_{{{s}_{k}}}}(\tau ) \right|\}}<\infty.
\end{equation}

As a result, the problem can be formulated as follows
\begin{equation}
\begin{aligned}
&P1:\mathop {{\rm{minimize}}}\limits_{{\bf{x}}(t),{\bf{p}}(t)} {\rm{ }}{\overline {PC} _{tot}}\\
&C1:{{\bar{Q}}_{{{s}_{k}}}}=\underset{t\to \infty }{\mathop{\lim }}\,\frac{1}{t}\sum\limits_{\tau =0}^{t-1}{E\{\left| {{Q}_{{{s}_{k}}}}(\tau ) \right|\}}<\infty,\\
&C2:\begin{array}{*{20}{l}}
{{\sum\limits_w {\sum\limits_{{s_k}} {x_u^{(k,w,{s_k})}(t)p_u^{(k,w,{s_k})}(t)} } } }\\
{{\rm{ + }}\sum\limits_l {\sum\limits_{{s_k}} {x_c^{(k,l,{s_j})}(t)p_c^{(k,l,{s_k})}(t)} }  \le {P_{total}}}
\end{array},\\
&C3: {\sum\limits_w {\sum\limits_{{s_k}} {x_u^{(k,w,{s_k})}(t)p_u^{(k,w,{s_k})}(t)} } } \le {P_u},\\
&C4:\sum\limits_{k \ne 0} {x_c^{(k,l,{s_k})}(t)p_c^{(k,l,{s_k})}(t)g_c^{(k,l,{s_0})}(t)}  \le {I_M},\\
&C5:\sum\limits_{{s_k}} {x_c^{(k,l,{s_j})}(t)}  \le 1,\sum\limits_{{s_k}} {x_u^{(k,w,{s_k})}(t)}  \le 1,\\
&C6: p_c^{(k,l,{s_k})}(t) \ge 0, p_u^{(k,w,{s_k})}(t) \ge 0,\\
&C7: x_c^{(k,l,{s_k})}(t) \in \{ 0,1\},x_u^{(k,w,{s_k})}(t) \in \{ 0,1\}.\\
\end{aligned}
\end{equation}
where $\{{p_c^{(k,l,{s_k})}(t)}\}$, $\{{p_u^{(k,w,{s_k})}(t)}\}$, $\{{x _c^{(k,l,{s_k})}(t)}\}$ and $\{{x _u^{(k,w,{s_k})}(t)}\}$ are variables. C1 is the queue stability constraint to guarantee all arrived data leaving the buffer in a finite time. C2 is the total transmission power constraint on both the licensed and unlicensed bands, while C3 is the transmission power constraint on the unlicensed bands due to the regulations \cite{3GPP}. C4 can restrict the interference arising from SBSs. C5 and C7 guarantee that each subcarrier of the SBS has been used at most by one user.

\section{An Online Energy-Aware Algorithm via Lyapunov Optimization}
We can exploit the drift-plus-penalty algorithm \cite{Yu2016} to solve the stochastic optimization problem P1. First, we introduce some necessary but pratical boundedness assumptions to derive the drift-plus-penalty expression of P1. We assume the following inequalities
\begin{equation}\label{boundness1}
\mathbb{E}\left\{ {{A_{{s_k}}}{{(t)}^2}} \right\} \le \psi ,{k \in {\cal K}\backslash \{ 0\} },\forall {s_k},
\end{equation}
\begin{equation}\label{boundness2}
\mathbb{E}\left\{ {{R_{{s_k}}}{{(t)}^2}} \right\} \le \psi ,{k \in {\cal K}\backslash \{ 0\} },\forall {s_k},
\end{equation}
hold for some finite constant $\psi$. In addition, ${{PC_{tot}}\left( t \right)}$ and ${{R_{tot}}\left( t \right)}$ are bounded respectively by
\begin{equation}\label{boundness3}
{P_{\min }} \le \mathbb{E}\left\{ {{PC_{tot}}\left( t \right)} \right\} \le {P_{\max }},
\end{equation}
\begin{equation}\label{boundness4}
{R_{\min }} \le \mathbb{E}\left\{ {{R_{tot}}\left( t \right)} \right\} \le {R_{\max }},
\end{equation}
where ${P_{\min }}$, ${P_{\max }}$, ${R_{\min }}$, ${R_{\max }}$ are some finite constants.
Define the Lyapunov function as \cite{Yu2016}
\begin{equation}\label{lyapunov function}
L\left( {{\bf{Q}}\left( {{t}} \right)} \right) = \frac{1}{2}\sum\limits_{k \in {\cal K}\backslash \{ 0\} } {\sum\limits_{{s_k}} {{{\left( {{Q_{{s_k}}}\left( t \right)} \right)}^2}} }.
\end{equation}
Then the one-slot conditional Lyapunov drift can be expressed as
\begin{equation}\label{lyapunov drift}
\Delta \left( {{\bf{Q}}\left( {{t}} \right)} \right) = \mathbb{E}\left\{ {L\left( {{\bf{Q}}\left( {{{t + 1}}} \right)} \right) - L\left( {{\bf{Q}}\left( {{t}} \right)} \right)|{\bf{Q}}}(t) \right\}.
\end{equation}

Thus, the drift-plus-penalty expression of P1 is defined as
\begin{equation}\label{drift-plus-penalty}
{V}\mathbb{E}\left( {{PC_{tot}}\left( t \right)|{\bf{Q}}}(t) \right) + \Delta \left( {{\bf{Q}}}(t) \right),
\end{equation}
where $V$ is a control parameter. The following lemma 1 provides the upper bound of the drift-plus-penalty expression.
\begin{lem}\label{lemma 1}
Assume link condition is i.i.d over slots. Under any power allocation algorithm, all parameter $V \ge 0$, and all possible queue length $\mathbf{Q}$, the drift-plus-penalty satisfies the following inequality:
\begin{equation}\label{inequality1}
\begin{aligned}
&{V}\mathbb{E}\left( {{PC_{tot}}\left( t \right)|{\bf{Q}}} \right)+ \Delta \left( {{\bf{Q}}} \right)\le C_0 + {\rm{V}}\mathbb{E}\left( {{PC_{tot}}\left( t \right)|{\bf{Q}}} \right) \\
& + \sum\limits_{k \in {\cal K}\backslash \{ 0\} } {\sum\limits_{{s_k}} {{Q_{{s_k}}}} \left( t \right)\left( {{A_{{s_k}}}\left( t \right) - {R_{{s_k}}}\left( t \right)|{\bf{Q}}} \right)}\\
\end{aligned}
\end{equation}
where $C_0$ is a positive constant, satisfying for all $t$
\begin{equation}\label{B}
C_0 \ge \frac{1}{2}\sum\limits_{k \in {\cal K}\backslash \{ 0\} } {\sum\limits_{{s_k}} {{\mathbb{E}}\left( {{A_{{s_k}}}{{\left( t \right)}^2}{\rm{ + }}{R_{{s_k}}}{{\left( t \right)}^2}|{\bf{Q}}} \right)} }.
\end{equation}
\end{lem}
\begin{proof}
Squaring both side of (\ref{queue}) and exploiting the inequality
\begin{equation}\label{inequality2}
{\left\{ {\max \left[ {Q - R} \right] + A} \right\}^2} \le {Q^2}{\rm{ + }}{R^2} + {A^2} - 2Q\left( {R - A} \right),
\end{equation}
we can get
\begin{equation}\label{inequality3}
\begin{aligned}
{\left[ {{Q_{{s_k}}}\left( {t + 1} \right)} \right]^2} & \le {\left[ {{Q_{{s_k}}}\left( t \right)} \right]^2}{\rm{ + }}{\left[ {{A_{{s_k}}}\left( t \right)} \right]^2}{\rm{ + }}{\left[ {{R_{{s_k}}}\left( t \right)} \right]^2}\\
&- 2{Q_{{s_k}}}\left( t \right)\left( {{R_{{s_k}}}\left( t \right) - {A_{{s_k}}}\left( t \right)} \right).\\
\end{aligned}
\end{equation}
Summarizing over ${s_k}$, we have
\begin{equation}\label{inequality4}
\small{\begin{array}{l}
\frac{{\sum\limits_{k \in {\cal K}\backslash \{ 0\} } {\left( {\sum\limits_{{s_k}} {{Q_{{s_k}}}{{\left( {t + 1} \right)}^2}}  - \sum\limits_{{s_k}} {{Q_{{s_k}}}{{\left( t \right)}^2}} } \right)} }}{2} \le \frac{{\sum\limits_{k \in {\cal K}\backslash \{ 0\} } {\sum\limits_{{s_k}} {\left( {{A_{{s_k}}}{{\left( t \right)}^2}{\rm{ + }}{R_{{s_k}}}{{\left( t \right)}^2}} \right)} } }}{2}\\
 - \sum\limits_{k \in {\cal K}\backslash \{ 0\} } {\sum\limits_{{s_k}} {{Q_{s_k}}\left( t \right)\left( {{R_{s_k}}\left( t \right) - {A_{s_k}}\left( t \right)} \right)} }
\end{array}}
\end{equation}
The left-hand-side of (\ref{inequality4}) equals to $\Delta \left( {{\bf{Q}}}(t) \right)$. Lemma 1 is proven.
\end{proof}

To push the objective P1 to its minimum, a proper power allocation algorithm is proposed to greedily minimize the drift-plus-penalty expression of P1. As a result, from the stochastic optimization theory, it is required to minimize the upper bound in (\ref{inequality1}) subject to the same constraints C2-C7 except the stability constraint C1. Therefore, the transformed problem P2 is given by

\begin{equation}
\begin{aligned}
P2: &{\min}~V \times {PC_{tot}}\left( t \right) - \sum\limits_{{s_k}} {{Q_{{s_k}}}} \left( t \right){R_{{s_k}}}\left( t \right)\\
&s.t. C2-C7.\\
\end{aligned}
\end{equation}
Unfortunately, the optimization is highly non-convex. Nevertheless, we can equivalently transform P2 to a D.C. program as discussed in the sequel.

For convenience's sake, we get rid of the slot index $t$ without ambiguity. It is noted that $\mathbf{x}$ is binary and the product term $\mathbf{x}\mathbf{p}$ is obviously non-convex, we can recast these constraints using the inequality $0 \le {\mathbf{p}} \le {\mathbf x}\Lambda$ \cite{Che2014}, where $\Lambda  > 0$ is a predefined constant. We can further transform the binary constraint C7 as the intersection of the following regions \cite{Cui2017a}
\begin{equation} \label{aaa40}
0 \le x_c^{(k,l,{s_j})} \le 1, 0 \le x_u^{(k,w,{s_k})} \le 1,
\end{equation}
\begin{equation} \label{aaa41}
\begin{array}{l}
\sum\limits_k {\sum\limits_l {\sum\limits_{{s_k}} {(x_c^{(k,l,{s_k})} - {{(x_c^{(k,l,{s_k})})}^2})} } } \\
 + \sum\limits_k {\sum\limits_w {\sum\limits_{{s_k}} {(x_u^{(k,w,{s_k})} - {{(x_u^{(k,w,{s_k})})}^2}) \le 0} } }.
\end{array}
\end{equation}

Although optimization variables $\mathbf{x}$ are continuous values, constraint (\ref{aaa41}) is non-convex. In order to deal with (\ref{aaa41}), we reformulate P2, as given by (\ref{long}), where $\lambda$ acts a penalty factor. It is proven that for sufficiently large values of $\lambda$, P3 can be equivalent to P2 \cite{Che2014}.
\newcounter{mytempeqncnt}
\begin{figure*}
\normalsize
\setcounter{mytempeqncnt}{\value{equation}}
\setcounter{equation}{33} 
\begin{equation} \label{long}
\begin{aligned}
P3: & {\rm{min}}V \times P{C_{tot}}\left( t \right) - \sum\limits_{{s_k}} {{Q_{{s_k}}}} \left( t \right){R_{{s_k}}}\left( t \right) + \lambda \sum\limits_k {\sum\limits_l {\sum\limits_{{s_k}} {(x_c^{(k,l,{s_k})} - {{(x_c^{(k,l,{s_k})})}^2})} } }  + \lambda \sum\limits_k {\sum\limits_w {\sum\limits_{{s_k}} {(x_u^{(k,w,{s_k})} - {{(x_u^{(k,w,{s_k})})}^2})} } } \\
s.t. & {\sum\limits_w {\sum\limits_{{s_k}} {p_u^{(k,w,{s_k})}} } }  + \sum\limits_l {\sum\limits_{{s_k}} {p_c^{(k,l,{s_k})}} }  \le {P_{total}},~ {\sum\limits_w {\sum\limits_{{s_k}} p_u^{(k,w,{s_k})}(t) } } \le {P_u},\\
&\sum\limits_{j \ne 0} {p_c^{(j,l,{s_j})}} g_c^{(j,l,{s_0})} \le {I_M},~p_c^{(k,l,{s_k})}(t) \le x_c^{(k,l,{s_k})}(t)\Lambda,~p_u^{(k,w,{s_k})}(t) \le x_u^{(k,w,{s_k})}(t)\Lambda,~ C5,~C6.\\
\end{aligned}
\end{equation}
\setcounter{equation}{\value{mytempeqncnt}}
\hrulefill 
\vspace*{4pt} 
\end{figure*}
\setcounter{equation}{34}
Define
\begin{equation}
\small{\begin{aligned}
 &f({\bf{P}},{\bf{x}}) =V \times P{C_{tot}} - \sum\limits_k {\sum\limits_l {\sum\limits_{{s_k}} {{Q_{{s_k}}}}} } \\
&B{\log _2}\left( {\sum\limits_k {p_c^{(k,l,{s_k})}g_c^{(k,l,{s_k})}}  + {\sigma ^2}} \right) - \sum\limits_k {\sum\limits_w {\sum\limits_{{s_k}} {{Q_{{s_k}}}R_u^{(k,w,{s_k})}} } } \\
&+ \lambda \sum\limits_k {\sum\limits_l {\sum\limits_{{s_k}} {(x_c^{(k,l,{s_k})})} } }  + \lambda \sum\limits_k {\sum\limits_w {\sum\limits_{{s_k}} {(x_u^{(k,w,{s_k})})} } },
\end{aligned}}
\end{equation}
\begin{equation}
\small{\begin{aligned}
g({\bf{P}},{\bf{x}}) &= - \sum\limits_k {\sum\limits_l {\sum\limits_{{s_k}} {{Q_{{s_k}}}B{{\log }_2}\left( {\sum\limits_{j \ne k} {p_c^{(j,l,{s_j})}g_c^{(j,l,{s_k})}}  + {\sigma ^2}} \right)} } } \\
 &+ \lambda \sum\limits_k {\sum\limits_l {\sum\limits_{{s_k}} {{{(x_c^{(k,l,{s_k})})}^2}}  + \lambda \sum\limits_k {\sum\limits_w {\sum\limits_{{s_k}} {{{(x_u^{(k,w,{s_k})})}^2}} } } } }
\end{aligned}}
\end{equation}

Since $f$ and $g$ are convex, the objective function is the difference of two convex functions, as given by $f-g$. As a result, P3 is a D.C. program. Therefore, we can apply successive convex approximation to obtain a local optimal solution of P3.

Let $i$ denote the iteration number. Since $g$ is convex, at the $i$-th iteration, we have
\begin{equation} \label{gyy}
\begin{array}{l}
f({\bf{p}},{\bf{x}}) - g\left( {{\bf{p}},{\bf{x}}} \right) \le f({\bf{p}},{\bf{x}}) - g\left( {{\bf{p}}(i - 1),{\bf{x}}(i{\rm{ - }}1)} \right)\\
 - {\nabla _{\bf{p}}}g\left( {{\bf{p}}(i - 1),{\bf{x}}(i{\rm{ - }}1)} \right)({\bf{p}} - {\bf{p}}(i - 1))\\
 - {\nabla _{\bf{x}}}g\left( {{\bf{p}}(i - 1),{\bf{x}}(i{\rm{ - }}1)} \right)({\bf{x}} - {\bf{x}}(i-1)),
\end{array}
\end{equation}
where ${{\bf{p}}(i-1)}$ and ${{\bf{x}}(i-1)}$ are the solutions of the problem at $(i-1)$-th iteration, and ${\nabla _{\bf{p}}}$ and ${\nabla _{\bf{x}}}$ are the gradient operation with respect to $\bf{p}$ and $\bf{x}$. As a result, P3 becomes a convex optimization problem, which can be efficiently solved by using standard convex optimization techniques, such as the interior-point method. Our proposed algorithm can be explicitly described in Algorithm 1.

\begin{algorithm}[t]
  \caption{Online Energy-Aware Spectrum Access and Power Allocation Algorithm}
  \label{Online Energy-Aware Resource Allocation Algorithm}
  \begin{algorithmic}[1]
  \STATE Initialize $\mathbf{p}{(0)}$ and $\mathbf{x}{(0)}$, and $t=0$.
  \STATE At the beginning of each slot $t$, acquire the current queue state $\mathbf{Q}(t)$ and the channel state $g_c^{(k,l,{s_k})}(t)$ and $g_u^{(k,w,{s_k})}(t)$, and obtain the number of Wi-Fi nodes $N_k(t)$ at SBS $k$.
  \REPEAT
  \STATE Optimize P3 to obtain optimal $\mathbf{p}(t)$ and $\mathbf{x}(t)$ by using a D.C. program.
  \UNTIL{convergence of $\mathbf{p}$ and $\mathbf{x}$}
  \end{algorithmic}
\end{algorithm}
\subsection{Performance analysis}
Since the system state per time slot is i.i.d., We can quantify the performance of our proposed algorithm, by means of Markovian randomness \cite{Yu2016}. Denote $PC_{tot}^*\left( t \right)$, $R_{{s_k}}^{\rm{*}}\left( t \right)$ as the optimal power consumption and the corresponding rate. If the boundness assumptions (\ref{boundness1})-(\ref{boundness4}) hold, there exists an i.i.d spectrum access and power allocation algorithm, satisfying
\begin{equation}\label{aa1}
\mathbb{E}\left( {R_{{s_k}}^*\left( t \right)} \right) \ge \mathbb{E}\left( {{A_{{s_k}}}\left( t \right)} \right) + \varepsilon,
\end{equation}
where $\varepsilon$ is a small positive value. The following Theorem reveals the performance bounds of average power and average delay of the proposed algorithm.
\begin{thm}\label{lemma 2}
Suppose the system state per slot time is i.i.d, the average power and average queue length of the proposed algorithm are bounded respectively by
\begin{equation}\label{tradeoff1}
\overline {{PC}}_{tot}  \le \frac{C_0}{V} + \overline {PC}_{tot}^*,
\end{equation}
\begin{equation}\label{tradeoff2}
\overline Q  \le \frac{{C_0 + V\overline {PC}_{tot}^* }}{\varepsilon },
\end{equation}
where $C_0$ and $\varepsilon$ are defined in (\ref{B}) and (\ref{aa1}), respectively.
\end{thm}

Its proof uses a standard result in the stochastic optimization theory \cite{Yu2016}. Theorem 1 implies a tradeoff of $[O(1/V),O(V)]$ between power consumption and queue length (i.e., delay). In other word, by increasing control parameter $V$, the power consumption can converge to the optimal value but the traffic delay gets increasing.

\section{Simulation}
We conduct the simulation with the time slot length to be $10~ms$, and run each experiment for 5000 slots. There are $K=3$ SBSs, which each has $L=2$ licensed subcarriers and $W=4$ unlicensed subcarriers. We set SBS users and Wi-Fi nodes are uniformly distributed. And the arrival data packet of each users follows Poisson distribution. The channel gains of licensed and unlicensed bands follow the Rayleigh fading. Set the power amplifier of licensed and unlicensed bands as ${1 \mathord{\left/ {\vphantom {1 {{\xi _c} = }}} \right.  \kern-\nulldelimiterspace} {{\xi _c} = }}{1 \mathord{\left/ {\vphantom {1 {{\xi _u} = }}} \right. \kern-\nulldelimiterspace} {{\xi _u} = }}0.35$. Let $P_{total}$ be $46~dBm$, and $P_u$ be $23~dBm$. Set $P{{C}_{idle}}= 1~W$ and $P{{C}_{\text{static}}} = 9~W$.

\begin{figure}[!t]
\centering
\includegraphics[width=3in]{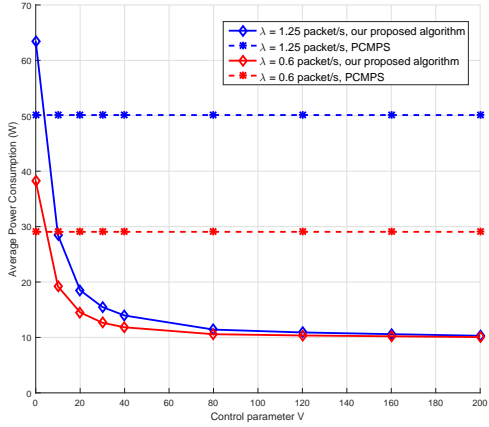}
\caption{Average power consumption versus control parameter $V$.}
\label{fig.1}
\end{figure}

\begin{figure}[!t]
\centering
\includegraphics[width=3in]{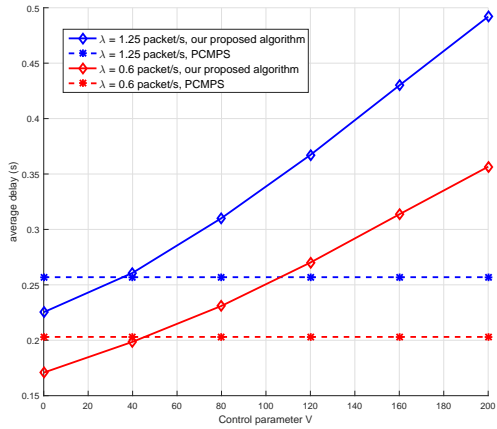}
\caption{Average delay versus control parameter $V$.}
\label{fig.1}
\end{figure}
We compare the proposed algorithm under different control parameter $V$ with a power consumption minimization per slot (PCMPS). The PCMPS minimizes the power consumption per slot, subject to C2-C7 and a new rate constraint ${R_{{s_k}}}(t) \ge {A_{{s_k}}}(t)$. The new constraint is added to guarantee the QoS of the users. In Fig. 2, we plot the total power consumption against $V$. It shows that when $V$ increases, the total power consumption of our proposed algorithm could decrease and converge to a point at the speed of $O(1/V)$ for any given traffic arrival rate $\lambda$. According to (\ref{tradeoff1}), the converged point is the optimal power consumption $\overline {PC}_{tot}^*$. And it is obviously observed that our proposed algorithm consumes less power than the PCMPS, when $V \geq 5$. This is because PCMPS ignores the queue states and always need to guarantee that the service rate is greater than arrival rates.
Fig. 3 shows the average traffic delay against $V$. As $V$ increases, the average traffic delay (or queue backlog) grows linearly in $O(V)$, which is consistent with (\ref{tradeoff2}).

Fig. 2 and Fig. 3 together show that we can achieve a tradeoff between power and delay. For example, if the network operator chooses $5 \leq V \leq 40$ for $\lambda {\rm{ = 1}}{\rm{.25}}$, the proposed algorithm outperforms the PCMPS in both the power and delay. In particular, the proposed algorithm can reduce the power consumption over PCMPS scheme by up to 72.1\% under the same traffic delay. A balance between the licensed channel interference and the unlicensed channel collision can also be achieved by the proposed algorithm.
\section{Conclusion}
In this paper, we have formulated a stochastic optimization to minimize the system average power consumption in the stochastic LAA-enabled SBSs and Wi-Fi networks, by jointly optimizing subcarrier assignment and power allocation between the licensed and unlicensed band. In the framework of Lyapunov optimization, an online energy-aware algorithm is developed. The theoretical analysis and simulation results show that our proposed algorithm can give a practical control and balance between power consumption and delay.

\section*{Acknowledgment}
The work was supported by National Nature Science Foundation of China Project (Grant No. 61471058), Hong Kong, Macao and Taiwan Science and Technology Cooperation Projects (2014DFT10320, 2016YFE0122900), the 111 Project of China (B16006) and Beijing Training Project for The Leading Talents in S\&T (No. Z141101001514026).
%


\end{document}